\documentclass{snkart}

\usepackage{mathtools}
\usepackage{tikz}

\let\complex\mathsf
\let\group\mathbf
\let\union\bigcup
\DeclareMathOperator{\End}{End}

\DeclareMathOperator{\Hom}{\complex{Hom}}
\DeclareMathOperator{\mhom}{mhom}
\DeclareMathOperator{\conv}{conv}
\newcommand{\nothing}{\mathord{\_}}
\newcommand{\geom}[1]{\lvert #1\rvert}   %

\newcommand{\posetA}{P} %
\newcommand{\face}{F}   %
\newcommand{\auto}{\phi} %
\let\georeal\geom

\newcommand{\arityf}{2}
\newcommand{\id}{1}
\hypersetup{%
  pdftitle  = {A topological proof of the Hell–Nešetřil dichotomy},
  pdfauthor = {S. Meyer and J. Opršal}}

\begin{document}

\title{A topological proof of the Hell–Nešetřil dichotomy}

\author{Sebastian Meyer}
\address{TU Dresden, Dresden, Germany}
\email{sebastian.meyer2@tu-dresden.de}

\author{Jakub Opršal}
\address{University of Birmingham, Birmingham, UK}
\email{jakub.oprsal@bham.ac.uk}

\thanks{The first author was funded by the European Research Council (Project POCOCOP, ERC Synergy Grant 101071674). Views and opinions expressed are however those of the author(s) only and do not necessarily reflect those of the European Union or the European Research Council Executive Agency. Neither the European Union nor the granting authority can be held responsible for them.
\\\medskip \copyright S.~Meyer, J.~Opršal; licensed under Creative Commons License CC-BY 4.0.}

\begin{abstract}
  We provide a new proof of a theorem of Hell and Nešetřil [J. Comb. Theory B, 48(1):92–110, 1990] using tools from topological combinatorics based on ideas of Lovász [J. Comb. Theory, Ser.~A, 25(3):319–324, 1978]. The Hell–Nešetřil Theorem provides a dichotomy of the graph homomorphism problem. It states that deciding whether there is a graph homomorphism from a given graph to a fixed graph $H$ is in P if $H$ is bipartite (or contains a self-loop), and is NP-complete otherwise. In our proof we combine topological combinatorics with the algebraic approach to constraint satisfaction problem.
\end{abstract}

\maketitle

\subsection*{Acknowledgements}

The first author thanks Manuel Bodirsky for pointing him towards this topic. The second author would like to thank Gianluca Tasinato for his immense patience in explaining the intrigues of algebraic topology. We both thank Roman Gundarin for noticing a mistake in an earlier version of this manuscript.

\section{Introduction}

In their seminal 1990 paper \cite{HellN90}, Hell and Nešetřil studied the complexity of the graph homomorphism problem, also known as the $H$-colouring problem.
The problem fixes a graph $H$, and asks whether there is a graph homomorphism (i.e., an edge-preserving map) from a given graph $G$ to $H$.
Graph homomorphisms naturally generalise graph colouring: a $k$-colouring is the same as a homomorphism to the $k$-element clique $K_k$.
The paper then provides the following complexity dichotomy assuming that $H$ is a finite graph \cite[Theorem~1]{HellN90}.

\begin{theorem} [Hell--Nešetřil] \label{thm:hell-nesetril}
    If $H$ is bipartite or contains a loop, then the $H$-coloring problem is in P.
    If $H$ is not bipartite and contains no loop then the $H$-coloring problem is NP-complete.
\end{theorem}

This theorem served as a foundation for the Feder--Vardi conjecture \cite{FederV98} which claimed that this dichotomy extends to homomorphism problems of finite relational structures over arbitrary signature; this general problem is known as the \emph{constraint satisfaction problem (CSP)}. The Feder--Vardi conjecture was confirmed by Bulatov \cite{Bulatov17} and Zhuk \cite{Zhuk20}.

In this paper, we provide a new proof of the Hell--Nešetřil theorem using topology.
Topological methods have recently emerged in the study of \emph{approximate graph colouring} problem which is a promise version of the $k$-colouring problem. In its simplest form, we are asked to find a $k$-colouring of a $3$-colourable graph for some (fixed) $k > 3$. Although some version of this problem is studied since 70's \cite{GareyJ76}, surprisingly little is known about the complexity of this problem. The best known NP-hardness is for colouring $3$-colourable graphs with $5$ colours \cite{BartoBKO21}, and the best known polynomial-time algorithm uses $\tilde O(n^{0.19747})$ colors \cite{KawarabayashiTY24}. Topological methods were first used in showing that $3$-colouring $H$-colourable graphs is NP-hard for all non-bipartite $3$-colourable graphs $H$ \cite{KrokhinO19,WrochnaZ20,KrokhinOWZ23}. Recently, similar methods were applied to provide NP-hardness of certain promise colouring of 3-uniform hypergraphs \cite{FilakovskyNOTW24}.
On the other hand, a connection between topology and complexity in the other direction was recently established by Schnider and Weber \cite{SchniderW24} who showed that Boolean CSPs exhibit a topological dichotomy along the same line as a complexity dichotomy provided earlier by Schaefer \cite{Schaefer78}. Our proof provides the tractable side of this dichotomy for graphs, and can be easily generalised to arbitrary (finite) relational structures. The full proof of the general dichotomy (including an earlier proof of the tractability side) was given recently by Meyer~\cite{Meyer24}.

We believe that our proof of the Hell--Nešetřil theorem is conceptually simpler then previously known proofs.
Unlike the original proof, which is purely combinatorial, it requires familiarity with some basics of algebraic topology and universal algebra.
We believe that our proof brings important insights into what makes a CSP tractable, and it brings to the light a new connection with homotopy theory.
In particular, it provides a new necessary condition for tractability of a (finite-template) CSP which is topological rather than algebraic.

\subsubsection*{A sketch of the proof}

Our proof builds on topological combinatorics which started with Lovász's result on the chromatic number of Kneser graphs \cite{Lovasz78}. In his paper, Lovász assigns to each graph a topological space, or more precisely a simplicial complex, called the \emph{neighbourhood complex}. Lovász then shows that if this space is $(k-2)$-connected then the graph cannot be coloured by $k$ colours, and that the neighbourhood complexes of Kneser graphs satisfy this property. In more recent retelling of the proof often a different simplicial complex is used, called the \emph{box complex} which has a natural action of the 2-element group.
We refer to a book by Matoušek \cite{Matousek03} for an accessible overview of these methods.

The second ingredient of the proof is a use of the \emph{algebraic approach to CSP} which connects the complexity of these problems to the algebraic structure of the \emph{polymorphisms} of the template. This general theory was developed over past few decades. Its foundations were established by Jeavons et al.~\cite{JeavonsCG97,BulatovJK05}, and we refer to Barto, Krokhin, and Willard \cite{BartoKW17} for a modern (although pre-dichotomy) overview of this theory.
Polymorphisms of a CSP template are multivariate endomorphisms, i.e., for $H$-colouring, a polymorphism is a homomorphism from $H^n$ to $H$ for some $n$. As consequence of a fundamental theorem of the algebraic approach, we get that $H$-colouring is NP-complete if $H$ does not have a \emph{Taylor polymorphism} \cite[Corollary~42]{BartoKW17}, which is a polymorphism $t$ of arity $n \geq 1$ that satisfies equations of the form:
\[
  t\begin{pmatrix}
    x & {*} & \dots & {*} \\
    {*} & x & \dots & {*} \\
    \vdots & & \ddots & \vdots \\
    {*} & {*} & \dots & x \\
  \end{pmatrix} =
  t\begin{pmatrix}
    y & {*} & \dots & {*} \\
    {*} & y & \dots & {*} \\
    \vdots & & \ddots & \vdots \\
    {*} & {*} & \dots & y \\
  \end{pmatrix}
\]
where $t$ is applied row-wise on each of the matrix, and the unspecified non-diagonal entries (denoted by $*$) are $x$'s or $y$'s in some fixed pattern.
Taylor \cite{Taylor77} introduced this equation in his study of homotopy of topological algebras as the algebraic condition which forces all topological algebras in the variety to have abelian fundamental group.

In the proof we connect the algebraic approach with Lovász's method. The hard implication of the Hell--Nešetřil dichotomy is to show that if $H$ is not bipartite and does not have a self-loop, then $H$-colouring is NP-hard.
We show that if $H$ is not bipartite and has a Taylor polymorphism then it has to have a self-loop, i.e., we provide a new proof of a statement proven earlier by Bulatov~\cite{Bulatov05} which is known to imply the Hell--Nešetřil theorem.
For simplicity, let us assume that $H$ is connected which may be done without loss of generality.
We then consider the box complex $B(H)$ of $H$, and observe that this complex is connected since $H$ is not bipartite.
The Taylor polymorphism of $H$ induces a continuous function $t^*\colon B(H)^n \to B(H)$. This function is enough to invoke Taylor's result to derive that the fundamental group is abelian. Our goal is nevertheless stronger. Namely, we will exploit the finiteness of $H$ to show that all the homotopy groups of $B(H)$ are actually trivial.

Triviality of homotopy groups of finite Taylor posets has been shown by Larose and Zádori \cite{LaroseZ05}. In essence, we use their result in the following way: Each simplicial complex has an associated partial order which is homotopy equivalent. If the simplicial complex is finite, so is its partial order. We then show that the Taylor polymorphism of $H$ is enough to invoke the theorem of Larose and Zádori. Although, unfortunately, it does not immediately induce a Taylor polymorphism of the poset, and hence we have to check that the proof applies in our case. This is the majority of the technical work in this paper.

Finally, after having proved that $B(H)$ is contractible, we invoke a generalisation of Brouwer's fixed-point theorem (a corollary of a theorem of Lefschetz \cite{Lefschetz37}), to derive that the action of the 2-element group has a fixed point. A simple observation then concludes that this fixed-point induces a self-loop on one of the vertices of $H$.

The proof presented in the subsequent sections is a formalisation of this argument. In the formal argument, we use a certain homomorphism complex instead of the box complex. This homomorphism complex makes it easier for us to deal with some technical intricacies of the proof, and it is homotopy equivalent to the box complex, and hence the core of the argument is the same.

\subsubsection*{Other proofs of the Hell--Nešetřil dichotomy}

An alternative proof of Hell--Nešetřil theorem was also provided by Bulatov \cite{Bulatov05}, and further simplifications were latter provided by Siggers \cite{Siggers10}.
The main motivation of Bulatov was to show that the $H$-colouring dichotomy follows the line of the algebraic dichotomy conjecture, i.e., that a non-bipartite graph $H$ with a Taylor polymorphism has to have a self-loop.
The difference from our proof is that this statement is proved using only combinatorics and algebra, and is more involved than ours.
Siggers's proof is obtained from the original proof of Hell and Nešetřil by providing new proofs to the few non-algebraic steps, and hence achieving the same algebraic dichotomy as Bulatov's.

Another way to prove the theorem is using cyclic polymorphisms: Barto and Kozik \cite{BartoK12} proved that Taylor polymorphism implies \emph{cyclic polymorphisms} from which it is relatively easy to show that the theorem follows. The difficulties of this proof are hidden in the proof of existence of the cyclic polymorphisms \cite[Theorem 4.2]{BartoK12} which constitutes a significant part of their paper.

Lastly, an analytical proof of the theorem have been provided by Kun and Szegedy \cite{KunS16}. This proof uses the analytical method more common in studying approximation of CSPs. The proof relies on a result of Dinur, Friedgut, and Regev \cite{DinurFR08} about independent sets in powers of non-bipartite graphs.

All of the proofs, with exception of the original proof, rely on the algebraic hardness condition due to Bulatov, Jeavons, and Krokhin \cite{BulatovJK05} stating that $H$-colouring is NP-hard unless $H$ has a Taylor polymorphism (see also Theorem~\ref{thm:bjk05} below). The proof of this theorem has also been refined several times \cite{BartoOP17, BartoBKO21}, and is generally considered to be well-understood at least within the algebraic theory of CSPs.

\section{Preliminaries}

We combine three well-known theories: the algebraic approach to the constraint satisfaction problem \cite{BartoKW17}, topological combinatorics \cite{Matousek03}, and algebraic topology \cite{Hatcher02}. The three cited sources give an accessible and detailed introduction to these topics, and we recommend to keep these sources at hand while reading this paper.
In this section, we outline some of the basic definitions, and known facts that will be useful in the proof. We include sketches of some proofs for reference.

We write $[n]$ for the set $\{1, \dots, n\}$, and $1_A$ for the identity function on a set $A$. We write nameless functions as $x\mapsto t(x)$ where $t(x)$ is an expression using $x$, and $\nothing \mapsto c$ denotes the constant function. All graphs and posets in this paper are finite.
The function $p_i \colon A^n \to A$ defined by $p_i(x_1, \dots, x_n) = x_i$ where $i \in [n]$ is called a \emph{projection} onto the $i$-th coordinate.

We use the terms endomorphism, automorphism, and isomorphism in the usual meaning, e.g., an isomorphism is a bijective homomorphism whose inverse is a homomorphism, etc.
If $\group G$ is an abelian group, the symbol $\End(\group G)$ denotes the endomorphism ring of $\group G$, i.e., the set of endomorphisms of $\group G$ with pointwise addition and composition as multiplication.

\subsection{Graphs and the \texorpdfstring{$H$}{H}-colouring problem}

We treat graphs as relational structures with one binary symmetric relation, i.e., a \emph{graph} is a pair $G = (V(G), E(G))$ where $V(G)$ is a set and $E(G) \subseteq V(G)^2$ such that if $(u,v) \in E(G)$ then $(v, u) \in E(G)$. Elements of $V(G)$ are called \emph{vertices} of $G$, and elements of $E(G)$ are called \emph{edges}.
A \emph{graph homomorphism} from a graph $G$ to a graph $H$ is a mapping $f\colon V(G) \to V(H)$ which preserves edges, i.e., it satisfies that $(f(u), f(v)) \in E(H)$ for all $(u, v) \in E(G)$.

We may now formally define the $H$-colouring problem which is a special case of a \emph{constraint satisfaction problem (CSP)}. For the sake of brevity, we won't define the CSP, and refer the reader to \cite{BartoKW17}.

\begin{definition} [$H$-colouring]
  Fix a graph $H$. The \emph{$H$-colouring problem} is the decision problem whose input is a finite graph $G$, and the goal is to decide whether $G$ maps homomorphically to $H$, or not.
\end{definition}

Two graphs $H$ and $H'$ are called \emph{homomorphically equivalent} if there are homomorphisms $f\colon H \to H'$ and $g\colon H' \to H$. A graph is a \emph{core} if it is not homomorphically equivalent to any of its proper subgraphs, i.e., each of its endomorphisms is an automorphism.
It may be observed that if $H$ and $H'$ are homomorphically equivalent then $H$-colouring and $H'$-colouring are identical problems.
For each finite graph $H$ there is a unique (up to isomorphism) core $H'$ which is homomorphically equivalent to $H$. We can therefore always assume that $H$ is a core.

Finally, a \emph{polymorphism of a graph} $H$ of arity $n$ is a mapping $f\colon H^n \to H$ that preserves edges in the sense that if $(u_i, v_i) \in E(H)$ for all $i \in [n]$, then
\[
  (f(u_1, \dots, u_n), f(v_1, \dots, v_n)) \in E(H).
\]
The core of the algebraic approach is that polymorphisms determine the complexity of $H$-colouring up to log-space reductions; if $H$ has enough polymorphisms, the $H$-colouring problem gets easier.

\subsection{Posets, simplicial complexes, and their geometric realization}

We refer to Matoušek \cite[Section 1.7]{Matousek03} for a detailed treatment of the relation between posets, simplicial complexes, and topological spaces used in this paper.

A \emph{partially ordered set (poset)} is a set $P$ together with a partial order on $P$ which we will denote by $\leq$. A partial order is a transitive, reflexive, and antisymmetric binary relation.

An element $p$ of a poset $P$ is \emph{irreducible} if it has either a unique upper cover, or a unique lower cover. A poset is \emph{ramified} if it does not contain any irreducible elements.
We say that a poset $P$ \emph{dismantles} to its subposet $Q$ if $Q$ is obtained from $P$ by iteratively removing an irreducible element. Every finite poset dismantles to a ramified subposet. As usual, we write $x < y$ if $x\le y$ and $x \neq y$.
A \emph{chain} in $P$ is a linearly-ordered subset of $P$, i.e., a set $\{x_1,\dots, x_k\}$ such that $x_1 \leq \dots \leq x_k$.
If $P$ and $Q$ are posets, we will use the notation $P^Q$ to denote the poset of all monotone maps $f\colon Q \to P$ ordered by pointwise comparison, i.e., $f \leq g$ if $f(x) \leq g(x)$ for all $x\in Q$.
Note that monotone maps $R \times Q \to P$ are in bijection with monotone maps $R \to P^Q$. There is a connection to ramified posets as shown in the following lemma. We include the proof for completeness.

\begin{lemma}[Larose and Zádori {\cite[Lemma~2.2]{LaroseZ97}}, based on Stong \cite{Strong66}]
    \label{lem:aloneComponent}
    A finite poset $P$ is ramified if and only if each connected component of $P^P$ that contains an automorphism of $P$ contains no other elements.
\end{lemma}

\begin{proof}
  Assume that $P$ is ramified. Let $f \in P^P$ be such that $f > \id_P$.
  Let $x$ be a maximal element satisfying $f(x) > x$. Since $f(x)$ is not an upper cover of $x$, there exists $y > x$ such that $y \not\ge f(x)$. This is in contradiction with $y = f(y)\ge f(x)$. Symmetrically, there is no $f < \id_P$, and hence the connected component of $\id_P$ contains no other elements. The general case of an arbitrary automorphism in place of $\id_P$ is analogous.

  Conversely, if $P$ is not ramified, there is $x$ with a cover $y$. Then the map which maps $x$ to $y$ and does not move any other element in the same component as the identity.
\end{proof}

A subposet $R$ of a poset $P$ is called a \emph{retract} if there is a monotone map $r\colon P \to R$ whose restriction to $R$ is identity. Equivalently, a retract $R$ is the image of $P$ under an endomorphism $r$ that satisfies $r^2 = r$.

A~\emph{(finite) simplicial complex} $\complex K$ is a finite downward-closed set of finite sets. The sets in $\complex K$ are called \emph{faces}, and the set $V(\complex K) = \union_{F\in \complex K} F$ is the set of \emph{vertices} of $\complex K$.
The \emph{order complex} of a poset $P$ is the simplicial complex whose vertices are the elements of $P$ and whose faces are all chains in $P$.
With every simplicial complex $\complex K$, one can associate a topological space $\lvert \complex K \rvert$, called the \emph{geometric realization} of $\complex K$, as follows:
Identify the vertex set of $\complex K$ with a set of points in general position in a sufficiently high-dimensional Euclidean space (here, general position means that the points in $F\cup G$ are affinely independent for all $F, G\in \complex K$). Then, in particular, the convex hull $\conv(F)$ is a geometric simplex for every $F\in \complex K$, and the geometric realization can be defined as the union $\lvert \complex K \rvert = \bigcup_{F\in \complex K} \conv(F)$ of these geometric simplices (see, e.g., \cite[Lemma~1.6.2]{Matousek03}).
In what follows, we blur the distinction between a simplicial complex and its geometric realization.

A \emph{geometric realization of a poset} $P$, denoted by $\lvert P \rvert$, is the geometric realization of its order complex. Intuitively, it is constructed by starting with points of $P$ with discrete topology, then connecting any two points $p < q$ by an arc, filling any three arcs connecting $p, q, r$ with $p < q < r$ with a triangle, etc. We will treat $P$ as a subset of $\lvert P \rvert$.
In this paper, we pay little attention to the intermediate simplicial complex, although we use terminology of \emph{faces} and \emph{vertices} coming from there, e.g., the vertices of $\geom P$ are the elements of $P$, and the faces of $\geom P$ are chains in $P$.

Every monotone function $f\colon P \to Q$ between two posets is a simplicial map between the two order complexes, and every such map consequently induces a continuous function $\geom f \colon \geom P \to \geom Q$ between the geometric realizations.
This map is defined as the linear extension of $f$ viewing the elements of $P$ and $Q$ as subsets of the corresponding geometric realizations.

\subsubsection{Homotopy and a fixed-point theorem}

We will define homotopy groups of posets through their geometric realization. This differs from how Larose and Zádori \cite{LaroseZ05} define homotopy groups of posets, which is through defining a topology directly on the poset itself. The resulting groups are nevertheless isomorphic \cite[Theorems 1\&2]{McCord66} (see also \cite[Theorem~2.3]{LaroseZ05}).

Our proof needs basic understanding of notions such as a homotopy, homotopy groups of a space, and homotopy equivalence. We recall some of these notions here, and refer to any textbook on algebraic topology, e.g., \cite{Hatcher02}, for an in depth exposition.

Informally, two continuous functions $f, g \colon X \to Y$ are \emph{homotopic}, if they can be continuously transformed to each other. This is formally expressed by a continuous map $H \colon X \times [0, 1] \to Y$ such that $H(x, 0) = f(x)$ and $H(x, 1) = g(x)$ for all $x\in X$. If $f$ and $g$ are homotopic, we will write $f \sim g$.
Two spaces $X$ and $Y$ are called \emph{homotopy equivalent} if there are continuous maps $f\colon X \to Y$ and $g\colon Y \to X$ such that $fg \sim 1_Y$ and $gf \sim 1_X$. A space $X$ is \emph{contractible} if it is homotopy equivalent to the singleton space.

The following lemma provides some intuition about homotopy of monotonous maps.

\begin{lemma} \label{lem:leq-homotopy}
  If $f, g\colon P \to Q$ are monotone and $f \leq g$, then $\geom f$ and $\geom g$ are homotopic.
\end{lemma}

\begin{proof}
  Consider the poset $P \times \{0, 1\}$ with the component-wise order. Then the map $H \colon P\times \{0, 1\}$ defined by $H(p, 0) = f(p)$ and $H(p, 1) = g(p)$ is monotone.
  Observe that $\geom {P \times \{0, 1\}}$ is homeomorphic to $\geom P \times [0, 1]$, hence $\geom H$ can be viewed as a map $H' \colon \geom P \times [0, 1] \to \geom Q$. This maps is the required homotopy from $\geom f$ to $\geom g$.
\end{proof}

In general, it may be observed that homotopy of monotone maps is a symmetric, transitive closure of the relation in the above lemma, i.e., $\geom f \sim \geom g$ if and only if there is a sequence of monotone maps $f_0, \dots, f_n$ such that $f = f_0 \leq f_1 \geq f_2 \leq \dots f_n = g$. In other words, $\geom f$ and $\geom g$ are homotopic if and only if $f$ and $g$ are in the same connected component of $P^P$.

Similarly, dismantlability gives homotopy equivalence, although it should be noted that the converse is not true; e.g., there are non-trivial ramified posets whose geometric realization is contractible.

\begin{lemma} \label{lem:dismantlability}
  If $P$ dismantles to $Q$ then $\geom P$ and $\geom Q$ are homotopy equivalent.
\end{lemma}

\begin{proof}
  It is enough to prove the statement in case $Q$ is obtained from $P$ by removing an irreducible element $\ell \in P$. Let us further assume that $\ell$ has a unique upper cover $u$; the other case is symmetric.
  Define a map $f\colon P \to Q$ by $f(\ell) = u$ and $f(q) = q$ for all $q\in Q$. The map $f$ is clearly monotone. Let $i\colon Q \to P$ be the inclusion.
  We claim that $\geom f$ and $\geom i$ witness the required homotopy equivalence.
  Indeed, we get that $fi = 1_Q$, and that $if \leq 1_P$. The claim then follows from the fact that $\geom {if} = \geom i\circ \geom f$ and Lemma~\ref{lem:leq-homotopy}.
\end{proof}

Homotopy groups of a topological space $X$ pointed in $x_0 \in X$ are denoted by $\pi_n(X, x_0)$; we refer to \cite[Section 4.1]{Hatcher02} for precise definition. The elements of the group $\pi_n(X, x_0)$ are homotopy classes of (pointed) continuous maps $S^n \to X$, where $S^n$ denotes the $n$-dimensional sphere.
If the space $X$ is path-connected then $\pi_n(X, x_0)$ does not depend on the choice of $x_0$ \cite[p.~342]{Hatcher02}, and we will write $\pi_n(X)$ in that case. Also recall that $\pi_n(X, x_0)$ is abelian for all $n \geq 2$ \cite[p.~340]{Hatcher02}.
Every continuous map $f\colon X\to Y$ induces a group homomorphism $f_* \colon \pi_n(X, x_0) \to \pi_n(Y, f(x_0))$, and moreover, if $f \sim g$, then $f_* = g_*$. If $X$ and $Y$ are connected and homotopy equivalent, then $\pi_n(X)$ is isomorphic to $\pi_n(Y)$ for all $n$.
We write simply $f_*$ for $\geom f_*$ if $f\colon P \to Q$ is a monotone map between posets $P$ and $Q$. Note that $f_* g_* = (fg)_*$.

We rely on a fixed-point theorem that is a generalisation of Brouwer's fixed-point theorem and is well-known in algebraic topology. It is the following corollary of the Lefschetz fixed-point theorem \cite{Lefschetz37}.

\begin{theorem} [a corollary of the Lefschetz fixed-point theorem]
  \label{thm:fixed-point}
  If $X$ is a contractible finite simplicial complex, then every continuous function $f\colon X \to X$ has a fixed point.
\end{theorem}

\begin{proof}
  Since $X$ is contractible, we have that the $n$-th homology group $H_n(X)$ is trivial for all $n\neq 0$, and $H_0(X) = \mathbb Z$. Moreover, $f_*\colon H_0(X)\to H_0(X)$ is the identity map. This concludes that the Lefschetz number of $f$ is $\tau(f) = 1$, hence \cite[Theorem~2C.3]{Hatcher02} applies.
\end{proof}

\subsubsection{Multihomomorphisms and homomorphism complexes}

The core of the proof is a translation of the problem from graphs to posets, and from poset to the homotopy of simplicial complexes (where we consider two functions equal if they are homotopic).
The first transition is in terms of multihomomorphisms that are assigned to a pair of graphs.

\begin{definition} [the poset of multihomomorphisms]
  Let $G, H$ be graphs. A \emph{multihomomorphism} from $G$ to $H$ is a function $f\colon V(G) \to 2^{V(H)} \setminus \{\emptyset\}$ such that, for all edges $(u, v) \in E(G)$, we have that
  \[
    f(u) \times f(v) \subseteq E(H).
  \]
  We denote the set of all such multihomomorphisms by $\mhom(G, H)$.

  There is a natural order on multihomomorphisms by component-wise comparison, i.e., $f \leq g$ if $f(u) \subseteq g(u)$ for all $u\in V(G)$.
\end{definition}

Every homomorphism can be treated as a multihomomorphism that maps each vertex to a singleton set. The multihomomorphisms can be composed in a similar fashion as homomorphisms, i.e., if $f\in \mhom(A, B)$ and $g\in \mhom(B, C)$, then $(g\circ f)(a) = \bigcup_{b\in f(a)} g(b)$ is a multihomomorphism from $A$ to $C$.
Moreover, this composition is monotone, i.e., if $f \leq f'$ and $g \leq g'$ then $g\circ f \leq g' \circ f'$.

\begin{definition}
  Let $G, H$ be two graphs. The \emph{homomorphism complex} from $G$ to $H$ is the order complex of the multihomomorphism poset $\mhom(G, H)$. We will denote this complex by $\Hom(G, H)$.
  We treat $\Hom(G, H)$ as a topological space.
\end{definition}

Our argument uses the homomorphism complex $\Hom(K_2, H)$ where $K_2$ denotes the 2-clique (i.e., an unoriented edge).
As is well-known, this complex is homotopy equivalent the box complex of $H$ \cite[p.~137]{Matousek03}.
It can be also described as follows: its vertices are (oriented) complete bipartite subgraphs of $H$ and its faces are chains thereof with respect to inclusion.
We denote the vertices of $K_2$ by $0$ and $1$ and identify a homomorphism $m \in  \hom(K_2, H)\subseteq \mhom(K_2,H)$ with the edge $(m(0),m(1))$ of $H$.
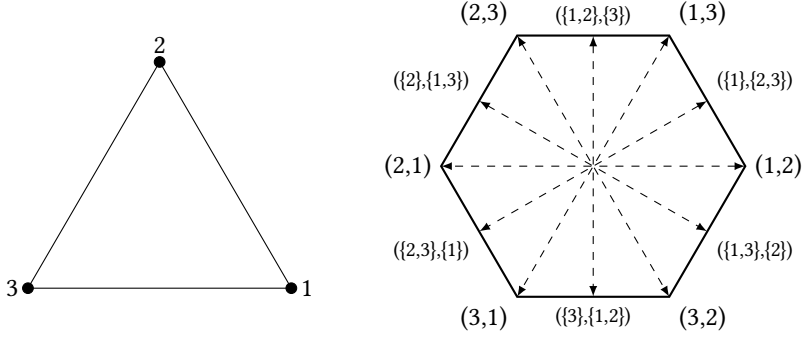
\begin{figure}
    \[
    \begin{tikzpicture}[baseline={([yshift=-.5ex]current bounding box.center)}]
        \draw (-30:2) node[anchor=west]{1}
        -- (90:2)    node[anchor=south]{2}
        -- (210:2)    node[anchor=east]{3}
        --  cycle;
        \filldraw (-30:2) circle [radius=2pt];
        \filldraw  (90:2) circle [radius=2pt];
        \filldraw (210:2) circle [radius=2pt];
    \end{tikzpicture}
    \qquad
    \begin{tikzpicture}[auto, swap, baseline={([yshift=-.5ex]current bounding box.center)}]
        \draw [thick] (0:2) node[anchor=west]{(1,2)}
        -- node (a) {\footnotesize (\{1\},\{2,3\})} (60:2)  node[anchor=south west]{(1,3)}
        -- node {\footnotesize (\{1,2\},\{3\})} (120:2) node[anchor=south east]{(2,3)}
        -- node {\footnotesize (\{2\},\{1,3\})} (180:2) node[anchor=east]{(2,1)}
        -- node (d) {\footnotesize (\{2,3\},\{1\})} (240:2) node[anchor=north east]{(3,1)}
        -- node {\footnotesize (\{3\},\{1,2\})} (300:2) node[anchor=north west]{(3,2)}
        -- node {\footnotesize (\{1,3\},\{2\})} cycle;
        \draw[dashed,latex-latex]  (0:2)--(180:2);
        \draw[dashed,latex-latex] (60:2)--(240:2);
        \draw[dashed,latex-latex](120:2)--(300:2);
        \draw[dashed,latex-latex] (30:1.73) to (210:1.73);
        \draw[dashed,latex-latex] (90:1.73) to (270:1.73);
        \draw[dashed,latex-latex] (150:1.73) to (330:1.73);
    \end{tikzpicture}
    \]
    \caption{The graph $K_3$ and the topological space $\Hom(K_2,K_3)$ which is homeomorphic to a circle. The multihomomorphisms are labelled and the action of $\auto$ on them is denoted by dashed lines.}
    \label{fig:HomSpace}
\end{figure}

This homomorphism complex has a natural action of $\mathbb Z_2$ induced by the non-trivial automorphism of $K_2$, $\nu\colon K_2 \to K_2$ which switches the two vertices. The mapping $\auto \colon \mhom(K_2, H) \to \mhom(K_2, H)$ defined as $\auto(m) = m \circ \nu$ is then a monotone bijection, and hence it induces a homeomorphism $\geom \auto \colon \Hom(K_2, H) \to \Hom(K_2, H)$.
On the simplicial level, this map flips the orientation of the complete bipartite subgraphs of $H$. We call the action \emph{flip}. See Figure~\ref{fig:HomSpace} for an example.
It is well-known that this action is (fixed-point) free as long as $H$ does not have a self-loop.
\begin{lemma} \label{lem:loop}
  Let $H$ be a graph. The flip $\georeal{\auto}$ on $\Hom(K_2, H)$ has a fixed point if and only if $H$ has a self-loop.
\end{lemma}

\begin{proof}
    If $H$ has a self-loop on $v$, then $\nothing \mapsto v$ is a homomorphism from $K_2$ to $H$, and hence a vertex of $\mhom(K_2, H)$, that is fixed by $\auto$.

    For the other direction assume that $\geom \auto$ has a fixed point which is an internal point of a face $\face$.
    Since $\georeal{\auto}$ is linear on faces, we get that $\face$ is invariant under $\auto$, consequently the minimal element $m$ of $\face$ is fixed by $\auto$ since $\auto$ is bijective and monotone. It is straightforward to see that any element of $m(0) = m(1)$ has to have a self-loop in $H$.
\end{proof}

We will also need a second lemma concerning $\auto$. Note that an edge $(u, v)\in E(H)$ defines a multihomomorphism from $K_2$ to $H$ defined as $0 \mapsto \{u\}, 1\mapsto \{v\}$. Below, we treat edges as such multihomorphisms.

\begin{lemma} \label{lem:edgeFlip}
    Let $H$ be non-bipartite. Then, it contains an edge $(u,v)$ that is in the same connected component as $(v,u) = \auto((u,v))$ in the poset $\mhom(K_2,H)$.
\end{lemma}
\begin{proof}
    Two edges $(u,v_1)$ and $(u,v_2)$ which start at the same point are connected in the poset since they have a common successor $0\mapsto \{u\}, 1\mapsto \{v_1,v_2\}$. Similarly, two incoming edges to a vertex $v$ are connected. Consequently, $(u_1, v_1)$ and $(u_2, v_2)$ are connected if there is an edge $(u_2,v_1)$ or, more generally, a path of odd length from $u_2$ to $v_1$. Since $H$ is non-bipartite, there is always a cycle of odd length connecting an edge with its inverse.
\end{proof}

\subsection{Taylor operations}

Taylor operations were introduced by Taylor \cite{Taylor77} while studying homotopy groups of topological algebras.

\begin{definition} \label{def:Taylor}
  Let $n\geq 1$. An $n$-ary operation $t\colon A^n \to A$ is \emph{Taylor} if it is idempotent, i.e., it satisfies $t(x, \dots, x) = x$ for all $x\in A$, and it satisfies, for all $i \in [n]$, an identity of the form
  \[
    t(x_{\alpha_i(1)}, \dots, x_{\alpha_i(n)}) =
    t(x_{\beta_i(1)}, \dots, x_{\beta_i(n)})
  \]
  for all $x_1, x_2 \in A$, where $\alpha_i, \beta_i\colon [n] \to [2]$ are fixed functions (that do not depend on $x_1$ and $x_2$) such that $\alpha_i(i) \neq \beta_i(i)$.
\end{definition}

It is easy to see that the above identities prevent this operation from being a projection.
We use the following algebraic criterion for NP-hardness of constraint satisfaction problems; see also \cite[Corollary~42]{BartoKW17}.

\begin{theorem} [Bulatov, Jeavons, and Krokhin \cite{BulatovJK05}]
  \label{thm:bjk05}
  If $H$ is a core, and $H$ does not have a Taylor polymorphism, then $H$-colouring is NP-complete.
\end{theorem}

One of key contributions of Taylors paper is the following theorem \cite[Corollary~5.3]{Taylor77}. If $\group G$ is a group, we say that an operation $f\colon G^n \to G$ is \emph{compatible}, or a \emph{group polymorphism}, if it commutes with the group operation, i.e., it is a group homomorphism $\group G^n \to \group G$.

\begin{lemma} [Taylor {\cite[Corollary~5.3]{Taylor77}}] \label{thm:taylor}
  Every group with a compatible Taylor operation is abelian.
\end{lemma}

\begin{proof}
  Assume that $\group G = (G, \cdot, {^{-1}}, 1)$ is a group with a compatible Taylor operation $t\colon \group G^n \to \group G$. For each $i\in [n]$ define
  \[
    N_i = \{ t(1, \dots, 1, \underset ix, 1, \dots, 1) \mid x \in G \}.
  \]
  Observe that $N_i$ is a subgroup for all $i$, and that, for all $x \in G$,
  \[
    x = t(x, \dots, x) = t(x, 1, \dots, 1) \cdot t(1, x, 1, \dots, 1) \cdots t(1, \dots, 1, x)
  \]
  from which we may derive that $G = N_1\cdot N_2 \cdots N_n$. Consequently, it is enough to prove the following claim.

  \begin{claim}
    For all $i\in [n]$, $[N_i, G] = 1$, i.e., $hg = gh$ for all $g\in G$ and $h \in N_i$.
  \end{claim}

  Without loss of generality assume $i = 1$.
  Since
  \[
    t(x_1, 1, \dots, 1)\cdot t(1, x_2, \dots, x_n) =
    t(x_1, x_2, \dots, x_n) =
    t(1, x_2, \dots, x_n)\cdot t(x_1, 1, \dots 1),
  \]
  we get $[N_1, \hat N] = 1$ where $\hat N = N_2\cdots N_n$ is the subgroup consisting of all elements of the form $t(1, x_2, \dots, x_n)$ for some $x_2, \dots, x_n \in G$.

  To conclude the claim, we will show that $G = \hat N$.
  Let $g\in G$ be an arbitrary element. Observe that
  \[
    t(x_1, \dots, x_n) = t(x_1, x_1, \dots, x_1) \cdot t(1, x_1^{-1}x_2, \dots, x_1^{-1}x_n) \in x_1\hat N
  \]
  for all $x_1, \dots, x_n \in G$. Now, substitute $g$ and $1$ for $x_1$ and $x_2$ into the first Taylor identity to get that
  \[
    t(g, x_2, \dots, x_n) = t(1, y_2, \dots, y_n)
  \]
  for some $x_2, \dots, x_n, y_2, \dots, y_n \in \{1, g\}$. As we observed above, $t(g, x_2, \dots, x_n) = gh$ for some $h\in \hat N$; while the right-hand side of the identity is an element $h' \in \hat N$. Consequently, $g = h'h^{-1}\in \hat N$.
  This concludes the proof of the claim and the lemma.
\end{proof}

Taylor used the theorem to show that for every pointed topological space $(X, x_0)$ with a continuous Taylor operation, $\pi_1(X, x_0)$ is abelian; see also \cite[Proposition 1.1]{Taylor77}. We will use a similar argument to derive that the fundamental group of a poset with a certain property is abelian.

\section{Taylor and contractibility of finite posets}

In this section, we present the technical contribution of this paper containing a mild generalisation of a theorem of Larose and Zádori \cite[Theorem~3.2]{LaroseZ05} which states that posets that allow a monotone Taylor operation have trivial homotopy groups.

Our generalisation uses the following weaker notion of a monotone Taylor operation. This weaker version is obtained by replacing Taylor identities with inequalities.

\begin{definition}
  Let $(P, \leq)$ be a poset. A \emph{sub-Taylor polymorphism} of arity $n$ of this poset is a monotone map $t\colon P^n \to P$, such that $t(x, \dots, x) \geq x$ for all $x\in P$, together with $n$ binary monotone maps $s_i\colon P^2 \to P$ that satisfy, for all $i \in [n]$,
  \begin{align*}
    t(x_{\alpha_i(1)}, \dots, x_{\alpha_i(n)}) &\geq s_i(x_1, x_2) \\
    t(x_{\beta_i(1)}, \dots, x_{\beta_i(n)}) &\geq s_i(x_1, x_2)
  \end{align*}
  for all $x_1, x_2 \in A$, where $\alpha_i, \beta_i\colon [n] \to [2]$ are fixed functions such that $\alpha_i(i) \neq \beta_i(i)$.
\end{definition}

Remark that if $t$ is a Taylor operation, then it is also a sub-Taylor operation since we may define $s_i$'s by $s_i(x_1, x_2) = t(x_{\alpha_i(1)}, \dots, x_{\alpha_i(n)})$.

We may now formulate our generalisation of \cite[Theorem~3.2]{LaroseZ05}.

\begin{theorem} \label{thm:larose-zadori}
  Any finite poset $P$ with a sub-Taylor polymorphism has trivial homotopy, i.e., it satisfies $\pi_d(\geom P, p_0) = 0$ for all $d\geq 1$ and $p_0 \in P$.
\end{theorem}

We will closely follow the proof of Larose and Zádori with a few changes.
Let us first state a few lemmata. The first lemma is a consequence of Lemma~\ref{thm:taylor}, and it is proved using Taylor's ideas \cite{Taylor77} together with the fact that two monotone functions satisfying $f \leq g$ induce the same function on homotopy; see also \cite[Section 3.4.1]{KrokhinOWZ23} or \cite[Appendix~C]{FilakovskyNOTW24}.

\begin{lemma} \label{lem:taylor}
  Any poset $P$ with a sub-Taylor polymorphism has abelian homotopy, i.e., $\pi_1(\geom P, p_0)$ is abelian for all $p_0\in P$.
\end{lemma}

\begin{proof}
  Assume that $t, s_1, \dots, s_n$ are sub-Taylor operations.
  We assume that $p_0$ is a maximal element and thus $t(p_0, \dots, p_0) = p_0$. This assumption is without loss of generality since $\pi_1(\geom P, p_0)$ only depends on the connected component of $p_0$.

  First, by a standard argument, it follows that a monotone operation $f\colon P^n \to P$ such that $f(p_0, \dots, p_0) = p_0$ induces a compatible operation $f_* \colon \pi_1(\geom P, p_0)^n \to \pi_1(\geom P, p_0)$ relying on the fact that $\pi_1(\geom {P^n}, (p_0, \dots, p_0))$ and $\pi_1(\geom P, p_0)^n$ are naturally isomorphic.
  Furthermore, by following the argument in \cite[Proposition 1.1]{Taylor77}, we get that if
  \[
    g(x_1, \dots, x_m) = f(x_{\mu(1)}, \dots, x_{\mu(n)})
  \]
  for some $\mu \colon [n] \to [m]$ and monotone functions $g\colon P^m \to P$ and $f\colon P^n \to P$, then also
  \[
    g_*(x_1, \dots, x_m) = f_*(x_{\mu(1)}, \dots, x_{\mu(n)}).
  \]
  Finally, observe that if two monotone functions $f, g \colon P^n \to P$ satisfy $f \leq g$, then $f_* = g_*$. This is since $\geom f \sim \geom g$ by Lemma~\ref{lem:leq-homotopy}, and hence the induced action on the fundamental group is identical.

  We claim that the function $t_* \colon \pi_1(\geom P, p_0)^n \to \pi_1(\geom P, p_0)$, induced by the sub-Taylor operation $t\colon P^n \to P$, is a compatible Taylor operation.
  It is a group homomorphism by definition. To show that it is idempotent, observe that the function $\tilde t \colon P \to P$ defined by $\tilde t(x) = t(x, \dots, x)$ satisfies $\tilde t \geq 1_P$, hence $\geom {\tilde t} \sim 1_{\geom P}$, and finally, $\tilde t_* = 1_{\pi_1(\geom P, p_0)}$. We therefore get that $t_*(x, \dots, x)  = \tilde t_*(x) = x$ as required.
  Furthermore, $t$ satisfies the $i$-th Taylor identity with the same distribution of variables as in the sub-Taylor inequalities, i.e., the same $\alpha_i$ and $\beta_i$. This is since the functions
  \begin{align*}
    \ell_i(x_1, x_2) &= t(x_{\alpha_i(1)}, \dots, x_{\alpha_i(n)}) \\
       r_i(x_1, x_2) &= t(x_{\beta_i(1)}, \dots, x_{\beta_i(n)})
  \end{align*}
  satisfy $\ell_i \geq s_i \leq r_i$, and hence $(\ell_i)_* = (r_i)_*$ by the same argument as above. Consequently,
  \[
    t_*(x_{\alpha_i(1)}, \dots, x_{\alpha_i(n)})
    = (\ell_i)_*(x_1, x_2)
    = (r_i)_*(x_1, x_2)
    = t_*(x_{\beta_i(1)}, \dots, x_{\beta_i(n)}).
  \]
  This concludes that $\pi_1(\geom P, p_0)$ has a compatible Taylor operation, and therefore is abelian by Lemma~\ref{thm:taylor}.
\end{proof}

Recall that if $d \geq 2$, then $\pi_d(X, x_0)$ is abelian for all spaces $X$ and $x_0 \in X$ \cite[p.~340]{Hatcher02}, and hence the above lemma implies that all the homotopy groups of a poset with a sub-Taylor polymorphism are abelian.

\begin{lemma} \label{lem:retract}
  If a poset $P$ has a sub-Taylor polymorphism, then so does each of its retracts $R$.
\end{lemma}

\begin{proof}
  Assume $r\colon P \to R$ is a retraction, and $t, s_1, \dots, s_n$ are the sub-Taylor monotone operations on $P$. It is straightforward to check that $rt, rs_1, \dots, rs_n$ are sub-Taylor operations of $R$. For example, for all $p \in R$, we have $r(p) = p$, and hence $rt(p, \dots, p) \ge r(p) = p$.
\end{proof}

For each operation $f\colon P^n\to P$, each $k \in [n]$, and each $a_1, \dots, a_{k-1}, a_{k+1}, \dots, a_n \in P$, we call the function $e \colon P \to P$ defined by
\[
    e(x) = f(a_1, \dots, a_{k-1}, x, a_{k+1}, \dots, a_n)
\]
a \emph{slice} of $f$. If $f$ is monotone then all its slices are monotone as well.
The following is proved by the same argument as \cite[Corollary~2.2]{LaroseZ05} and \cite[Lemma 2.9]{Larose06}.

\begin{lemma} \label{lem:pouzet}
  Let $P$ be a finite, connected, ramified poset. Let $f\colon P^{n+1} \to P$ be an idempotent monotone map. Fix $a_1, \dots, a_n \in P$ and consider the slice $e\colon P \to P$ defined by $e(x) = f(a_1, \dots, a_n, x)$.
  Then either 
  \begin{itemize}
      \item the slice $e$ is not onto, or
      \item the map $f$ is the projection on the last component, i.e., $f(y_1, \dots, y_n, x)=e(x)=x$ for all $y_1, \dots, y_n \in P$.
  \end{itemize} 
\end{lemma}

\begin{proof}
    Note that the two cases exclude each other.
  Consider the map $g\colon P^n \to P^P$ defined by $g(y_1, \dots, y_n) = x \mapsto f(y_1, \dots, y_n, x)$. Since $P$ is connected, so is the image of $g$. If $e = g(a_1, \dots, a_n)$ is onto, then it is an automorphism of $P$. By Lemma~\ref{lem:aloneComponent}, $e$ is alone in its component of $P^P$. Consequently, $g$ is a constant map $\nothing \mapsto e$, and hence $f$ does not depend on its first $n$ coordinates. As $f$ is idempotent, $e$ is the identity and $f$ is the projection.
\end{proof}

\begin{lemma} \label{lem:projection-contractible}
  If $X$ is a topological space such that two projections $p_i, p_j\colon X^n \to X$ for some $i\neq j$ are homotopic, then $X$ is contractible.
\end{lemma}

\begin{proof}
  Without loss of generality, assume $i=1$ and $j=2$, and let $H\colon X^n \times [0, 1] \to X$ be a homotopy between the two projections, i.e., 
  \[
    H(x, y, z_3,\dots,z_{n}; 0) = x \text{ and } H(x, y,z_3,\dots,z_{n}; 1) = y
  \]
  for all $x, y, z_3,\dots,z_n \in X$.
  Pick $y_0 \in X$ arbitrarily, and define $H' \colon X \times [0, 1] \to X$ by $H'(x, t) = H(x, y_0, y_0, \dots, y_0; t)$. Then $H'$ is continuous since $H$ is, and moreover it is a homotopy between $1_X$ and the constant $y_0$ map, hence a contraction of $X$ to $y_0$.
\end{proof}

It is also possible to state this result for posets in the same language as Lemma~\ref{lem:aloneComponent}: Consider a poset $P$ and two projections $p_i, p_j\colon P^n\to P$ where $i\neq j$. If they are in the same connected component in $P^{P^n}$, then $P$ is contractible.

The final essential piece to our proof is the following lemma. Part of the proof is heavily inspired by \cite[Theorem 2]{Larose91}, although we add a few additional assumptions to simplify the proof.
We present the full argument here for completeness.

\begin{lemma} \label{lem:larose-thm2}
  Let $P$ be a non-contractible, connected poset.
  If, for some $m\geq 2$, $P$ admits a monotone idempotent operation $g\colon P^m \to P$ which is not a projection, then it also admits a binary monotone idempotent operation that is not a projection.
\end{lemma}
\begin{proof}
  We prove the statement by induction on $m \geq 2$. There is nothing to prove for $m = 2$. For the induction step, assume that $m \geq 2$ and $g\colon P^{m+1} \to P$ is a monotone idempotent operation which is not a projection.
  Consider the operation $h\colon P^2 \to P$ defined by
  \[
    h(x, y) = g(x, y, \dots, y).
  \]
  Clearly, $h$ is idempotent, monotone, and binary. If it is not a projection, we are done. Otherwise, we consider two cases:
  \begin{itemize}
    \item $h(x, y) = y$. For each $a\in P$, consider the operations
      \[
        g_a(x_1, \dots, x_m) = g(a, x_1, \dots, x_m),
      \]
      all of which are idempotent since $h(a, y) = y$. Note that $g_a$ is monotone since $a\le a$.
      We claim that there exists $a\in P$ such that $g_a$ is not a projection.
      Otherwise, each $g_a$ is a projection, and since $P$ is connected and $a \mapsto g_a$ is a monotone operation $P \to P^{P^m}$, all $g_a$'s are homotopic. Consequently, by Lemma~\ref{lem:projection-contractible}, they are projection on the same coordinate, which implies that $g$ itself is a projection.
      This concludes that $g_a$ is an idempotent monotone operation which is not a projection and is of smaller arity than $g$.
    \item $h(x, y) = x$. Consider the operation
      \[
        h'(x, y) = g(y, x, y, \dots, y).
      \]
      Again, we may assume that $h'$ is a projection. We claim that $h'(x, y) = y$. This is since there are two elements $a < b$ in $P$, and
      \[
        b = h(b, a) = g(b, a, a, \dots, a) \leq g(b, a, b, \dots, b) = h'(a, b),
      \]
      which implies that $h'(a, b) \neq a$.
      After flipping the two first coordinates of $g$, we can continue as in the first case. \qedhere
  \end{itemize}
\end{proof}

Finally, we get to the proof of the main theorem of this section.

\begin{proof}[Proof of Theorem~\ref{thm:larose-zadori}]
  Fix $d \geq 1$. We will prove the statement by induction on the size of $P$. In the usual fashion, we will formulate this induction by assuming that $P$ is the smallest counterexample and derive a contradiction.
  From the minimality of $P$ and Lemma~\ref{lem:retract}, we may immediately derive that:
  \begin{itemize}
    \item $P$ is connected, hence $\pi_d(\geom P, p_0)$ does not depend on the choice of $p_0$. Otherwise a connected component of a $p_0$ with $\pi_d(\geom P, p_0) \neq 0$ is a retract and thus a smaller poset with a sub-Taylor polymorphism.
    \item $P$ is ramified. If $P$ is not ramified, it may be dismantled to a proper retract $R$ which is homotopy equivalent. The poset $R$ would again yield a smaller counterexample.
    \item $\geom P$ is not contractible. Otherwise $\pi_d(\geom P) = 0$ for all $d$.
  \end{itemize}

  We claim that $P$ admits a binary monotone idempotent operation $f\colon P^2 \to P$ which is not a projection. This operation can be constructed from the sub-Taylor operations $t, s_i$.
  First, we show that $t$ is idempotent.
  We have
  \[
    t(x, \dots, x) \geq x
  \]
  for all $x \in P$. Consequently, $x \mapsto t(x, \dots, x)$ is in the same connected component of $P^P$ as the identity. 
  Since $P$ is ramified this implies that $t(x, \dots, x) = x$ by Lemma~\ref{lem:aloneComponent}.
  Second, we claim that $t$ is not a projection. This is best argued by contradiction: If $t$ was the $i$-th projection, we would have
  \[
    x_{\alpha_i(i)}
    = t(x_{\alpha_i(1)}, \dots, x_{\alpha_i(n)})
    \leq s_i(x_1, x_2)
    \geq t(x_{\alpha_i(1)}, \dots, x_{\alpha_i(n)})
    = x_{\beta_i(i)},
  \]
  and, since $\alpha_i(i) \neq \beta_i(i)$, the two projections $\geom P^2 \to \geom P$ would be homotopic and $P$ would be contractible by Lemma~\ref{lem:projection-contractible}.
  We may then conclude that $P$ admits an idempotent monotone binary operation $f$ which is not a projection by Lemma~\ref{lem:larose-thm2}.

  Observe that Lemma~\ref{lem:pouzet} implies that no slices of $f$ are onto: If a slice $e(x) = f(a, x)$ is not onto, then $f(y, x) = e(x)$ for all $y$ by the lemma. In particular, $e(x) = f(x, x) = x$, and hence $f$ is the second projection. The other case is symmetric.

  Consider the homomorphism $f_* \colon \pi_d(\geom P)^\arityf \to \pi_d(\geom P)$, and observe that it is idempotent since $f$ is (see also the proof of Lemma~\ref{lem:taylor}). Furthermore, $\pi_d(\geom P)$ is abelian either by Lemma~\ref{lem:taylor}, if $d = 1$, or since higher homotopy groups are abelian, if $d \geq 2$. We will write the group operation of $\pi_d(\geom P)$ additively.
  Let $r_1(u) = f_*(u, 0)$, $r_2(u) = f_*(0, u)$, and observe that this defines elements of $\End(\pi_d(\geom P))$.
  Since $f_*$ is a group homomorphism, we have
  \[
    f_*(u_1, u_2)
    = f_*(u_1, 0) + f_*(0, u_2)
    = r_1(u_1) + r_2(u_2).
  \]
  Moreover, since $f_*$ is idempotent, we get that
  \(
    u = f_*(u, u) = r_1(u) + r_2(u)     
  \)
  for all $u \in \pi_d(\geom P)$, and hence $r_1 + r_2 = 1$.

  \begin{claim}
    Both $r_1$ and $r_2$ are nilpotent.
  \end{claim}
  \begin{proof}[Proof of the claim]
    We prove that $r_1$ is nilpotent, the nilpotency of $r_2$ then follows by symmetry.
    Consider any fixed $a_2 \in P$ and $e\colon P \to P$ defined by
    \[
      e(x) = f(x, a_2).
    \]
    Observe that $r_1(u) = f_*(u, 0) = e_*(u)$ for all $u$.
    This is since $\nothing \mapsto a_2$ is a trivial loop in $\pi_d(\geom P)$. Consequently $r_1 = e_*$.

    Since $P$ is finite, there exist $k \geq 1$ such that $e^k = e^{2k}$, and hence $e^k$ is a retraction. Its image is a proper retract $R$ of $P$ since $e$ is not onto (it is a slice of $f$).
    Finally, $(e_*)^k = (e^k)_* = 0$ since $R$ has a sub-Taylor polymorphism by Lemma~\ref{lem:retract}, and hence $\pi_d(\geom R) = 0$ since $P$ is a minimal counterexample.
  \end{proof}

  We claim that $r_2$ is invertible. This is since $r_1$ is nilpotent, hence there is $k$ such that $r_1^k = 0$, and consequently,
  \[
    r_2 (1 + r_1 + r_1^2 + \dots + r_1^{k-1}) =
    (1 - r_1)(1 + r_1 + r_1^2 + \dots + r_1^{k-1}) =
    1 - r_1^k = 1.
  \]
  Finally, since any power of an invertible element is invertible, we get that $0$ is invertible which is only possible in the trivial ring.
\end{proof}

The above theorem states that a finite poset with a sub-Taylor polymorphism has the same homotopy as a finite discrete set. Consequently, by combining this with the Whitehead theorem, we get that the geometric realization of such a poset is homotopy equivalent to a discrete set.

\begin{corollary} \label{cor:sub-TaylorContractible}
  If $P$ is a finite poset with a sub-Taylor polymorphism, then every connected component of $\lvert P \rvert$ is contractible.
\end{corollary}

\begin{proof}
  Let $C = \pi_0(\lvert P\rvert)$ be the set of connected components of $P$, and let $f\colon \lvert P \lvert \to C$ be the map that maps each point to its connected component. The map $f$ then induces isomorphisms on homotopy since $f_*\colon \pi_0(\lvert P \rvert) \to \pi_0(C)$ is bijective by definition, and $f_*\colon \pi_d(\lvert P\lvert, x_0) \to \pi_d(C, f(x_0))$ is trivial for all $d \geq 1$.
  By the Whitehead theorem \cite[Theorem~4.5]{Hatcher02} we get that $f$ is a homotopy equivalence, and hence has a homotopy inverse $g$. In particular, each component $c \in C$ contracts to $g(c)$.
\end{proof}

Although we do not need it in the proof, let us mention here that using the close connection between posets and simplicial complexes, we may obtain a version of Theorem~\ref{thm:larose-zadori} for simplicial complexes which says that finite simplicial complexes can either have non-trivial homotopy, or non-trivial algebraic properties but not both. See also \cite[Theorem~5.17]{Meyer24}.

\begin{corollary}
    If $\complex A$ is a finite simplicial complex with a compatible simplicial Taylor operation $f\colon \complex A^n\to \complex A$, then every connected component of $\complex A$ is contractible.
\end{corollary}

\begin{proof}
    Consider the face poset $P$ of $\complex A$ and observe that the monotone function $f'\colon P^n \to P$ defined by
    \[
    f' (F_1, \dots, F_n) = \{f (v_1, \dots, v_n) \mid v_1 \in F_1, \dots, v_n \in F_n \}
    \]
    is a sub-Taylor polymorphism; we prove this in more detail in an analogous case in Lemma~\ref{lem:mhom-sub-Taylor} below.
    By Corollary~\ref{cor:sub-TaylorContractible}, every component of $\geom P$ is contractible. Since $\geom P$ is the barycentric subdivision of $\complex A$, and hence it is homotopy equivalent to $\complex A$, the same holds for $\complex A$.
\end{proof}

\section{Proof of the Hell–Nešetřil theorem}
\label{section:HellNesetril}

In this section, we prove Theorem~\ref{thm:hell-nesetril}. We start with proving the following lemma that allows us to use Corollary~\ref{cor:sub-TaylorContractible}.

\begin{lemma} \label{lem:mhom-sub-Taylor}
  If a graph $H$ has a Taylor polymorphism, then $\mhom(G, H)$ has a sub-Taylor polymorphism for all graphs $G$.
\end{lemma}
\begin{proof}
  Every polymorphism $f\colon H^n\to H$ induces a monotone map $f'\colon \mhom(G, H)^n \to \mhom(G, H)$ defined by
  \[
      f'(m_1,\dots,m_n) = (v \mapsto \{f(h_1, \dots, h_n)\mid h_i\in m_i(v) \text{ for all $i\in [n]$} \}).
  \]
  Let $t\colon H^n \to H$ be a Taylor polymorphism, and let $s_1, \dots, s_n\colon H^2\to H$ be maps defined by
  \[
    s_i(x_1, x_2)
    = t(x_{\alpha_i(1)}, \dots, x_{\alpha_i(n)})
    = t(x_{\beta_i(1)}, \dots, x_{\beta_i(n)})
  \]
  for all $i \in [n]$, where $\alpha_i$ and $\beta_i$ are the corresponding functions witnessing that $t$ is Taylor.
  We claim that $t'$, $s'_1$, \dots, $s'_n$ are sub-Taylor operations witnessed by the same $\alpha_i$'s and $\beta_i$'s:
  Let $i \in [n]$, $m_1$ and $m_2$ be arbitrary multihomomorphisms, and $v \in V(G)$. Then
  \begin{align*}
    t(m_{\alpha_i(1)}, \dots, m_{\alpha_i(n)})(v)
    &= \{ t(h_1, h_2, \dots, h_n) \mid
          h_1 \in m_{\alpha_i(1)}(v), \dots, h_n \in m_{\alpha_i(n)}(v) \} \\
    &\supseteq \{ t(h_{\alpha_i(1)}, \dots, h_{\alpha_i(n)}) \mid
                  h_1 \in m_1(v), h_2 \in m_2(v) \} \\
    &= \{ s(h_1, h_2) \mid
          h_1 \in m_1(v), h_2 \in m_2(v) \} \\
    &= s'_i(m_1, m_2)(v)
  \end{align*}
  and thus $t'(x_{\alpha_i(1)}, \dots, x_{\alpha_i(n)})\ge s'_i(x, y)$. The other two inequalities $t'(x_{\beta_i(1)}, \dots, y_{\beta_i(n)})\ge s'_i(x, y)$ and $t'(x, \dots, x)\ge x$ are proven similarly.
\end{proof}

Next we provide a new proof of the following theorem which was first obtained by Bulatov~\cite{Bulatov05}.

\begin{theorem} \label{thm:loop}
    Every non-bipartite graph $H$ that has a Taylor polymorphism has a self-loop.
\end{theorem}
\begin{proof}
    Consider the poset $\posetA\coloneqq \mhom(K_2,H)$.
    By Lemma~\ref{lem:mhom-sub-Taylor}, $\posetA$ has a sub-Taylor polymorphism and by Corollary~\ref{cor:sub-TaylorContractible}, every connected component of $\geom{\posetA}$ is contractible.

    By Lemma~\ref{lem:edgeFlip}, there is an edge $(u,v)$ which is connected to
    $\auto((u,v))$ in $\posetA$ and thus also in $\georeal{\posetA}$. Since $\georeal{\auto}$ respects connected components of $\georeal{\posetA}$ (as any continuous function does), it restricts to an automorphism of the connected component of $(u,v)$.
    Therefore, we can apply Theorem~\ref{thm:fixed-point} on this component and get that $\georeal{\auto}$ has a fixed point.
    Now, $H$ has a loop by Lemma~\ref{lem:loop}.
\end{proof}

We may now conclude the Hell--Nešetřil theorem.

\begin{proof}[Proof of Theorem~\ref{thm:hell-nesetril}]
Let $H$ be a graph. We may assume that $H$ is a core or replace it with an equivalent core graph \cite[Theorems~16–17]{BartoKW17}. We distinguish the following cases:
\begin{description}
    \item[$H$ is bipartite] Then, $H$ is empty, nonempty without edges, or bipartite with edges. In the first two cases, a graph admits an $H$-colouring if and only if its vertex-set, or edge-set is empty, respectively. In the third case, a graph admits an $H$-colouring if and only if it is bipartite (this is since a graph is bipartite if and only if it allows a homomorphism to an edge, and an edge maps homomorphically to $H$).
    Either of these decisions is in P.
    \item[$H$ has a loop] Then every graph admits an $H$-colouring, and we can decide this in constant time.
    \item[$H$ is not bipartite and has no self-loop]
    By Theorem~\ref{thm:loop}, $H$ has no Taylor polymorphism. Therefore, $H$-colouring is NP-complete by Theorem~\ref{thm:bjk05}. \qedhere
\end{description}
\end{proof}

\section{Conclusion}

We discuss a few consequences of our results, and possible directions of future research and applications in the complexity of more general CSPs. 
We refer to \cite[Section~2]{BartoKW17} for definitions of CSPs and relational structures, and to \cite[Section~2.3]{FilakovskyNOTW24} for a general definition of the homomorphism complex.

Firstly, unlike previous proofs of the Hell--Nešetřil theorem, our proof does not rely on any specific graph-theoretical observations, and hence these methods can be used to provide a new general necessary criterion for tractability of finite-template CSPs.
In particular, we can prove that solution spaces of tractable finite-template CSPs are homotopy equivalent to discrete sets in the following sense.

\begin{corollary} \label{cor:necessary}
    Let $G$ and $H$ be two finite relational structures over the same signature such that $H$ has a Taylor polymorphism. Then, every connected component of $\Hom(G,H)$ is contractible.
\end{corollary}

\begin{proof}
  The proof in the general case is essentially identical to the proof for graphs. In particular, we may argue as in Lemma~\ref{lem:mhom-sub-Taylor} to show that $\mhom(G,H)$ allows a monotone sub-Taylor operation, which gives the required by Corollary~\ref{cor:sub-TaylorContractible}.
\end{proof}

This corollary can be extended to a full topology dichotomy of all finite-template CSPs along the lines of Schnider and Weber \cite{SchniderW24}; the hardness side and more details of the tractability side are presented by Meyer~\cite{Meyer24}.

Secondly, and similarly to proofs of Bulatov \cite{Bulatov05} and Siggers \cite{Siggers10}, our proof has some algebraic consequences. In particular, it can be used in place of Bulatov's or Siggers' proof in the proof that every finite structure with a Taylor polymorphism also has a 6-ary \emph{Siggers polymorphism} which satisfies the equations $s(x, y, z, x, y, z) = s(y, x, z, x, z, y)$. Furthermore, our method appears to be well-suited to provide other similar algebraic consequences. For example, can we provide a new topological proof that finite structures with Taylor polymorphisms have \emph{cyclic terms}?

Finally, a natural continuation of this work is a generalisation from CSPs to promise CSPs. In particular, we may ask: `How does the necessary condition from tractability of finite-template CSPs (provided by Corollary~\ref{cor:necessary}) generalise to finite-template promise CSPs?'
 
\bibliographystyle{alphaurl}

\begin{thebibliography}{KOW{\v{Z}}23}

\bibitem[BBKO21]{BartoBKO21}
Libor Barto, Jakub Bul\'{\i}n, Andrei Krokhin, and Jakub Opr\v{s}al.
\newblock Algebraic approach to promise constraint satisfaction.
\newblock {\em J. {ACM}}, 68(4):28:1--66, 8 2021.
\newblock \href {http://arxiv.org/abs/1811.00970} {arXiv:1811.00970}, \href
  {https://doi.org/10.1145/3457606} {doi:10.1145/3457606}.

\bibitem[BJK05]{BulatovJK05}
Andrei Bulatov, Peter Jeavons, and Andrei Krokhin.
\newblock Classifying the complexity of constraints using finite algebras.
\newblock {\em SIAM J. Comput.}, 34(3):720--742, 2005.
\newblock \href {https://doi.org/10.1137/S0097539700376676}
  {doi:10.1137/S0097539700376676}.

\bibitem[BK12]{BartoK12}
Libor Barto and Marcin Kozik.
\newblock Absorbing subalgebras, cyclic terms, and the constraint satisfaction
  problem.
\newblock {\em Log. Methods Comput. Sci.}, 8(1), 2012.
\newblock \href {https://doi.org/10.2168/LMCS-8(1:7)2012}
  {doi:10.2168/LMCS-8(1:7)2012}.

\bibitem[BKW17]{BartoKW17}
Libor Barto, Andrei Krokhin, and Ross Willard.
\newblock Polymorphisms, and how to use them.
\newblock In Andrei Krokhin and Stanislav {\v Z}ivn{\' y}, editors, {\em The
  Constraint Satisfaction Problem: Complexity and Approximability}, volume~7 of
  {\em Dagstuhl Follow-Ups}, pages 1--44. Schloss Dagstuhl -- Leibniz-Zentrum
  für Informatik, Dagstuhl, Germany, 2017.
\newblock \url{https://drops.dagstuhl.de/opus/frontdoor.php?
  source\_opus=6959}, \href {https://doi.org/10.4230/DFU.Vol7.15301.1}
  {doi:10.4230/DFU.Vol7.15301.1}.

\bibitem[BOP17]{BartoOP17}
Libor Barto, Jakub Opršal, and Michael Pinsker.
\newblock The wonderland of reflections.
\newblock {\em Israel Journal of Mathematics}, 223(1):363–398, November 2017.
\newblock \href {https://doi.org/10.1007/s11856-017-1621-9}
  {doi:10.1007/s11856-017-1621-9}.

\bibitem[Bul05]{Bulatov05}
Andrei~A. Bulatov.
\newblock {$H$}-coloring dichotomy revisited.
\newblock {\em Theor. Comput. Sci.}, 349(1):31--39, 2005.
\newblock \href {https://doi.org/10.1016/J.TCS.2005.09.028}
  {doi:10.1016/J.TCS.2005.09.028}.

\bibitem[Bul17]{Bulatov17}
Andrei~A. Bulatov.
\newblock A dichotomy theorem for nonuniform {CSP}s.
\newblock In {\em 2017 IEEE 58th Annual Symposium on Foundations of Computer
  Science ({FOCS})}, pages 319--330, Berkeley, CA, USA, 2017. IEEE.
\newblock \href {https://doi.org/10.1109/FOCS.2017.37}
  {doi:10.1109/FOCS.2017.37}.

\bibitem[DFR08]{DinurFR08}
Irit Dinur, Ehud Friedgut, and Oded Regev.
\newblock Independent sets in graph powers are almost contained in juntas.
\newblock {\em Geometric and Functional Analysis}, 18(1):77--97, January 2008.
\newblock \href {https://doi.org/10.1007/s00039-008-0651-1}
  {doi:10.1007/s00039-008-0651-1}.

\bibitem[FNO{\etalchar{+}}24]{FilakovskyNOTW24}
Marek Filakovsk{\'{y}}, Tamio{-}Vesa Nakajima, Jakub Opr\v{s}al, Gianluca
  Tasinato, and Uli Wagner.
\newblock Hardness of linearly ordered 4-colouring of 3-colourable 3-uniform
  hypergraphs.
\newblock In {\em 41st International Symposium on Theoretical Aspects of
  Computer Science, {STACS} 2024}, volume 289 of {\em LIPIcs}, pages
  34:1--34:19. Schloss Dagstuhl -- Leibniz-Zentrum f{\"{u}}r Informatik, 2024.
\newblock \href {https://doi.org/10.4230/LIPIcs.STACS.2024.34}
  {doi:10.4230/LIPIcs.STACS.2024.34}.

\bibitem[FV98]{FederV98}
Tom{\'{a}}s Feder and Moshe~Y. Vardi.
\newblock The computational structure of monotone monadic {SNP} and constraint
  satisfaction: A study through datalog and group theory.
\newblock {\em SIAM Journal on Computing}, 28(1):57--104, 1998.
\newblock \href {https://doi.org/10.1137/S0097539794266766}
  {doi:10.1137/S0097539794266766}.

\bibitem[GJ76]{GareyJ76}
M.~R. Garey and David~S. Johnson.
\newblock The complexity of near-optimal graph coloring.
\newblock {\em J. {ACM}}, 23(1):43--49, 1976.
\newblock \href {https://doi.org/10.1145/321921.321926}
  {doi:10.1145/321921.321926}.

\bibitem[Hat02]{Hatcher02}
Allen Hatcher.
\newblock {\em Algebraic topology}.
\newblock Cambridge University Press, Cambridge, 2002.

\bibitem[HN90]{HellN90}
Pavol Hell and Jaroslav Nešetřil.
\newblock On the complexity of {$H$}-coloring.
\newblock {\em J. Comb. Theory {B}}, 48(1):92--110, 1990.
\newblock \href {https://doi.org/10.1016/0095-8956(90)90132-J}
  {doi:10.1016/0095-8956(90)90132-J}.

\bibitem[JCG97]{JeavonsCG97}
Peter Jeavons, David~A. Cohen, and Marc Gyssens.
\newblock Closure properties of constraints.
\newblock {\em Journal of the ACM}, 44(4):527--548, 1997.
\newblock \href {https://doi.org/10.1145/263867.263489}
  {doi:10.1145/263867.263489}.

\bibitem[KO19]{KrokhinO19}
Andrei~A. Krokhin and Jakub Opr\v{s}al.
\newblock The complexity of 3-colouring {H}-colourable graphs.
\newblock In {\em 60th {IEEE} Annual Symposium on Foundations of Computer
  Science, {FOCS} 2019}, pages 1227--1239, Baltimore, Maryland, USA, 2019.
  {IEEE} Computer Society.
\newblock \href {https://doi.org/10.1109/FOCS.2019.00076}
  {doi:10.1109/FOCS.2019.00076}.

\bibitem[KOW{\v{Z}}23]{KrokhinOWZ23}
Andrei~A. Krokhin, Jakub Opr\v{s}al, Marcin Wrochna, and Stanislav
  {\v{Z}}ivn{\'{y}}.
\newblock Topology and adjunction in promise constraint satisfaction.
\newblock {\em {SIAM} J. Comput.}, 52(1):38--79, 2023.
\newblock \url{https://doi.org/10.1137/20m1378223}, \href
  {https://doi.org/10.1137/20M1378223} {doi:10.1137/20M1378223}.

\bibitem[KS16]{KunS16}
G{\'{a}}bor Kun and Mario Szegedy.
\newblock A new line of attack on the dichotomy conjecture.
\newblock {\em Eur. J. Comb.}, 52:338--367, 2016.
\newblock \href {https://doi.org/10.1016/J.EJC.2015.07.011}
  {doi:10.1016/J.EJC.2015.07.011}.

\bibitem[KTY24]{KawarabayashiTY24}
Ken{-}ichi Kawarabayashi, Mikkel Thorup, and Hirotaka Yoneda.
\newblock Better coloring of 3-colorable graphs.
\newblock In {\em Proceedings of the 56th Annual {ACM} Symposium on Theory of
  Computing, {STOC} 2024}, pages 331--339. {ACM}, 2024.
\newblock \href {http://arxiv.org/abs/2406.00357} {arXiv:2406.00357}, \href
  {https://doi.org/10.1145/3618260.3649768} {doi:10.1145/3618260.3649768}.

\bibitem[Lar91]{Larose91}
Benoit Larose.
\newblock Finite projective ordered sets.
\newblock {\em Order}, 8(1):33--40, 1991.
\newblock \href {https://doi.org/10.1007/bf00385812} {doi:10.1007/bf00385812}.

\bibitem[Lar06]{Larose06}
Benoit Larose.
\newblock Taylor operations on finite reflexive structures.
\newblock {\em International Journal of Mathematics and Computer Science},
  1(1):1--26, 2006.
\newblock \url{https://future-in-tech.net/R-Larose.pdf}.

\bibitem[Lef37]{Lefschetz37}
S.~Lefschetz.
\newblock On the fixed point formula.
\newblock {\em Annals of Mathematics}, 38(4):819--822, 1937.
\newblock \href {https://doi.org/10.2307/1968838} {doi:10.2307/1968838}.

\bibitem[Lov78]{Lovasz78}
L{\'{a}}szl{\'{o}} Lov{\'{a}}sz.
\newblock Kneser's conjecture, chromatic number, and homotopy.
\newblock {\em J. Comb. Theory, Ser. {A}}, 25(3):319--324, 1978.
\newblock \href {https://doi.org/10.1016/0097-3165(78)90022-5}
  {doi:10.1016/0097-3165(78)90022-5}.

\bibitem[LZ97]{LaroseZ97}
Benoit Larose and László Zádori.
\newblock Algebraic properties and dismantlability of finite posets.
\newblock {\em Discrete Mathematics}, 163(1):89--99, 1997.
\newblock \href {https://doi.org/10.1016/0012-365X(95)00312-K}
  {doi:10.1016/0012-365X(95)00312-K}.

\bibitem[LZ05]{LaroseZ05}
Benoit Larose and László Zádori.
\newblock Finite posets and topological spaces in locally finite varieties.
\newblock {\em Algebra Universalis}, 52(2–3):119–136, January 2005.
\newblock \href {https://doi.org/10.1007/s00012-004-1819-7}
  {doi:10.1007/s00012-004-1819-7}.

\bibitem[Mat03]{Matousek03}
Ji\v{r}\'i Matou\v{s}ek.
\newblock {\em Using the {B}orsuk-{U}lam Theorem}.
\newblock Lectures on Topological Methods in Combinatorics and Geometry.
  Springer-Verlag Berlin Heidelberg, 2003.
\newblock \href {https://doi.org/10.1007/978-3-540-76649-0}
  {doi:10.1007/978-3-540-76649-0}.

\bibitem[McC66]{McCord66}
Michael~C. McCord.
\newblock Singular homology groups and homotopy groups of finite topological
  spaces.
\newblock {\em Duke Math. J.}, 33:465--474, 1966.
\newblock \url{http://projecteuclid.org/euclid.dmj/1077376525}.

\bibitem[Mey24]{Meyer24}
Sebastian Meyer.
\newblock A dichotomy for finite abstract simplicial complexes.
\newblock preprint, 2024.
\newblock \href {http://arxiv.org/abs/2408.08199} {arXiv:2408.08199}, \href
  {https://doi.org/10.48550/arXiv.2408.08199} {doi:10.48550/arXiv.2408.08199}.

\bibitem[Sch78]{Schaefer78}
Thomas~J. Schaefer.
\newblock The complexity of satisfiability problems.
\newblock In {\em Proceedings of the 10th Annual {ACM} Symposium on Theory of
  Computing}, pages 216--226. {ACM}, 1978.
\newblock \href {https://doi.org/10.1145/800133.804350}
  {doi:10.1145/800133.804350}.

\bibitem[Sig10]{Siggers10}
Mark~H. Siggers.
\newblock A new proof of the {$H$}-coloring dichotomy.
\newblock {\em {SIAM} J. Discret. Math.}, 23(4):2204--2210, 2010.
\newblock \href {https://doi.org/10.1137/080736697} {doi:10.1137/080736697}.

\bibitem[Sto66]{Strong66}
R.~E. Stong.
\newblock Finite topological spaces.
\newblock {\em Trans. Amer. Math. Soc.}, 123:325--340, 1966.
\newblock \href {https://doi.org/10.2307/1994660} {doi:10.2307/1994660}.

\bibitem[SW24]{SchniderW24}
Patrick Schnider and Simon Weber.
\newblock A topological version of {S}chaefer's dichotomy theorem.
\newblock In {\em 40th International Symposium on Computational Geometry, SoCG
  2024}, volume 293 of {\em LIPIcs}, pages 77:1--77:16. Schloss Dagstuhl --
  Leibniz-Zentrum f{\"{u}}r Informatik, 2024.
\newblock \href {https://doi.org/10.4230/LIPIcs.SoCG.2024.77}
  {doi:10.4230/LIPIcs.SoCG.2024.77}.

\bibitem[Tay77]{Taylor77}
Walter Taylor.
\newblock Varieties obeying homotopy laws.
\newblock {\em Can. J. Math.}, XXIX(3):498--527, 1977.
\newblock \href {https://doi.org/10.4153/CJM-1977-054-9}
  {doi:10.4153/CJM-1977-054-9}.

\bibitem[W{\v{Z}}20]{WrochnaZ20}
Marcin Wrochna and Stanislav {\v{Z}}ivn{\'y}.
\newblock Improved hardness for {$H$}-colourings of {$G$}-colourable graphs.
\newblock In {\em Proceedings of the Fourteenth Annual {ACM-SIAM} Symposium on
  Discrete Algorithms ({SODA}'20)}, pages 1426--1435. SIAM, 2020.
\newblock \href {https://doi.org/10.1137/1.9781611975994.86}
  {doi:10.1137/1.9781611975994.86}.

\bibitem[Zhu20]{Zhuk20}
Dmitriy Zhuk.
\newblock A proof of the {CSP} dichotomy conjecture.
\newblock {\em J. {ACM}}, 67(5):30:1--30:78, 2020.
\newblock \href {https://doi.org/10.1145/3402029} {doi:10.1145/3402029}.

\end{thebibliography}
\newcommand{\etalchar}[1]{$^{#1}$}

\end{document}